\documentclass[sigconf]{acmart}

\usepackage[english]{babel}
\usepackage[T1]{fontenc}

\usepackage{amsmath}
\usepackage{amssymb}
\usepackage{booktabs}
\usepackage{graphicx}
\usepackage{subfig}

\usepackage{hyperref}
\usepackage{natbib}

\usepackage{balance}

\usepackage{array}
\newcolumntype{L}[1]{>{\raggedright\let\newline\\\arraybackslash\hspace{0pt}}m{#1}}
\newcolumntype{C}[1]{>{\centering\let\newline\\\arraybackslash\hspace{0pt}}m{#1}}
\newcolumntype{R}[1]{>{\raggedleft\let\newline\\\arraybackslash\hspace{0pt}}m{#1}}

\usepackage{multirow}

\usepackage{tikz}
\usetikzlibrary{shapes,decorations,arrows,calc,arrows.meta,fit,positioning}

\theoremstyle{plain}

\newtheorem{lemma}{Lemma}

\theoremstyle{definition}
\newtheorem{definition}{Definition}

\usepackage{url}
\newcommand\note[1]{\textcolor{red}{[#1]}}
\DeclareMathOperator*{\argmax}{arg\,max}
\newcommand{\myComment}[1]{} 

\setcopyright{none}

\renewcommand\footnotetextcopyrightpermission[1]{} 

\acmDOI{10.475/123_4}

\acmISBN{123-4567-24-567/08/06}

\acmConference[JCDL'19]{ACM/IEEE-CS Joint Conference on Digital Libraries}{June 2019}{Urbana-Champaign, Illinois, USA}
\acmYear{2019}
\copyrightyear{2019}

\acmPrice{15.00}

\acmSubmissionID{123-A12-B3}

\begin{document}
\title[Influence Dispersion Trees]{Go Wide, Go Deep: Quantifying the Impact of Scientific Papers through Influence Dispersion Trees}

\author{Dattatreya Mohapatra$^1$, Abhishek Maiti$^1$, Sumit Bhatia$^2$  and Tanmoy Chakraborty$^1$}
\affiliation{%
  \institution{$^1$IIIT-Delhi, India; $^2$IBM Research AI, New Delhi, India}
}
\email{{dattatreya15021,abhishek16005,tanmoy}@iiitd.ac.in, sumitbhatia@in.ibm.com}

\renewcommand{\shortauthors}{Mohapatra et al.}

\begin{abstract}
Despite a long history of use of `citation count' as a measure to assess the impact or influence of a scientific paper, the evolution of follow-up work inspired by the paper and their interactions through citation links have rarely been explored to quantify how the paper enriches the depth and breadth of a research field. We propose a novel data structure, called Influence Dispersion Tree (IDT) to model the organization of follow-up papers and their dependencies through citations. We also propose the notion of an ideal IDT for every paper and show that an ideal (highly influential) paper should increase the knowledge of a field vertically and horizontally. Upon suitably exploring the structural properties of IDT (both theoretically and empirically), we derive a suite of metrics, namely Influence Dispersion Index (IDI), Normalized Influence Divergence (NID) to quantify the influence of a paper. Our theoretical analysis shows that an ideal IDT configuration should have equal depth and breadth (and thus minimize the NID value).

We establish the superiority of NID as a better influence measure in two experimental settings. First, on a large real-world bibliographic dataset, we show that NID outperforms raw citation count as an early predictor of the number of new citations a paper will receive within a certain period after publication.  Second, we show that NID is superior to the raw citation count at identifying the papers recognized as highly influential through `Test of Time Award' among all their contemporary papers (published in the same venue). We conclude that in order to quantify the influence of a paper, along with the total citation count, one should also consider how the citing papers are organized among themselves to better understand the influence of a paper on the research field. For reproducibility, the code and datasets used in this study are being made available to the community.

\end{abstract}

\maketitle

\section{Introduction}
\label{sec:introduction}
A common consensus among the Scientometrics community is that the total number of citations received by a scientific article can be used to quantify its impact on the research field~\cite{garfield1972citation,garfield1964science}. Citation count, being a simple metric to compute and interpret, is commonly used in many decision-making processes such as faculty recruitment, fund disbursement, and tenure decisions. Many improvements over raw citation count have also been proposed by incorporating additional constraints. Examples include normalizing citation counts by the maximum citation count a paper could achieve in a particular research field~\cite{radicchi2008universality}, metrics inspired by PageRank~\cite{ding2009pagerank}, taking into
account the locations of citation mentions in the paper (e.g. Introduction, Related Work, etc.)~\cite{singh2015role}, understanding the reasons behind citations and assigning different weights to different citations based
on these reasons~\cite{chakraborty2016all}. 

While improvements over the raw citation count, these measures are fundamentally also \emph{aggregate} measures as they ignore the relationships between different (citing) papers that cite a given paper. We posit that such connections are useful and studying them can help us better understand the propagation of influence from a paper to its different citing papers.  Rather than proposing yet another variant of citation count,  we are interested in unraveling these structural connections between the set of followup papers of a given paper and understand the differentiating structural properties of influential papers.

\noindent {\bf Motivation:} We posit that the impact of a scientific paper can broadly be studied across two dimensions -- \emph{(i)} how many different research directions it gives rise to; and \emph{(ii)} how much traction these individual research directions gather in the field. In the former case, we say that the influence of the paper has \emph{breadth} and it helps in expanding the field horizontally, leading to an increase in the breadth of the field. A paper with such a broad influence may even trigger the emergence of a new sub-field. In the latter case, we say that the paper has had a \emph{deep} influence on the field with a large number of papers in a given research direction. Intuitively, \emph{highly influential papers are the ones that have a deep, and broad influence on the field}. Influence measures that are variants of the raw citation count of the paper may not offer such fine-grained understanding of the contribution of a paper to its field. Quantifying the impact of a paper in terms of its depth and breadth may also help to uncover the relationship between its different citing papers~\cite{huang2018number} and thus, understand the diffusion patterns of scientific ideas through citation links~\cite{chen2004tracing}, predict the structural virality~\cite{goel2015structural} and  citation cascade~\cite{min2017quantifying,huang2018number,abs-1806-00089}. While there have been recent efforts to study these structural properties of networks formed by a paper and its citing papers~\cite{min2017quantifying,huang2018number}, none of these studies have attempted to develop a metric to quantify the influence of a paper from its network topology. \emph{We are the first to propose a series of metrics to quantify a new facet of influence that a paper has had on its followup papers}.

\noindent {\bf Our Contributions:} Our major contributions are threefold:\\
\noindent
{\bf (i) A framework to model the depth and breadth of the influence of a paper} by a novel network structure, called the \emph{Influence Dispersion Tree (IDT)} (Section \ref{sec:approach}). The IDT of a paper $P$ is a directed tree rooted at $P$ with all its citing papers as the children. The tree is constructed such that the citing papers having citation links among themselves are grouped to represent a body of work influenced by the root paper $P$ (Section~\ref{sec:idt-construction}). These \emph{bodies} of work along with the number of papers in each group are then used to model the depth and breadth of impact of $P$.  We also present a theoretical analysis of the properties of the IDT structure and show how these properties are related to the citation count of the paper (Section~\ref{sec:idt-properties}). 

\noindent
{\bf (ii) A series of measures to quantify the influence of a scientific paper:} For a scholarly paper $P$, we propose a novel metric, called {\em Influence Dispersion Index} (IDI) derived from its IDT to quantify the contribution of the paper to its field (by increasing depth or breadth or both) through influence diffusion (Section \ref{sec:idi}). We argue that in an ideal scenario, the influence of a paper should be dispersed to maximize the depth as well as the breadth of its influence. We then derive the configuration of the IDT of such a paper and prove that such an optimal IDT configuration will have equal depth and breadth (and is equal to $\left\lceil\sqrt{n}\right\rceil$, where $n$ is the number of citations of a given paper). Next, we propose another metric, called {\em Influence Divergence} (ID) that measures how the IDI value of a paper diverges from IDI value of the optimal IDT configuration (Section~\ref{sec:nid}). A lower value of divergence indicates that the influence of the paper under consideration is dispersed in a way that is similar to that of the ideal case, and consequently, higher is the chance for the paper to be considered as a highly influential paper. We further derive a normalized version of ID, and call it {\em Normalized Information Divergence} (NID) that normalizes influence divergence values for different papers with different citation counts in the range $[0,1]$ and allows for comparing different papers based on their NID values. 

\noindent
 {\bf (iii) Empirical validation on large real-world datasets:} We use a large bibliographic dataset consisting of about $3.9$ million articles (Section \ref{sec:data}) to study the properties of the proposed IDT structure and test the effectiveness of proposed influence metrics. We construct IDTs for all the papers in the dataset and their analysis reveals several interesting observations (Section \ref{sec:general}). First, we observe that with an increase in the citation count, breadth of an IDT tends to grow much faster than the depth. The maximum value of breadth ($4,892$) is much higher than that of depth ($48$).  We infer that acquiring more citations over time often leads to an increase in the breadth instead of growth of an existing branch. Next, we find that the NID value decreases with an increase in citation count. This finding strengthens our hypothesis that the IDT of an highly influential paper tends to reach its optimal configuration by enhancing both the depth and the breadth of its research field. Third, we show that NID outperforms raw citation count as an early predictor to forecast the number of future citations a paper will receive (Section \ref{sec:future_citation_prediction}). Finally, we manually curate a set of 40 papers recognized as the most influential papers by their communities through `Test of Time' or `10 years influential paper' awards. Once again, we find that NID outperforms the raw citation count in identifying these influential papers (Section \ref{sec:tot_papers}). Most importantly, NID also provides an explanation why a paper has received such a prestigious award -- it is not only the number of followup papers (or citation count) that matters, but the factor which affects most is the way the followup papers are organized and linked in an IDT. In other words, {\em a highly influential paper tends to have an IDT with high breadth as well as high depth}. For reproducibility, the code and the dataset are available at \url{https://github.com/LCS2-IIITD/influence-dispersion}.
 
\section{Related work}
\label{sec:related-work}
There has been a plethora of research to measure the impact of scientific articles through various forms of citation analysis. In this section, we separate the related work into two parts -- (i) studies dealing with citation count and its variants for measuring the impact,  and (ii) studies  exploring detailed orchestration of citations around scientific papers.

\subsection{Citation Count as Impact Measure}
Searching for accurate and reliable indicators of research performance has a long and often controversial history. 
Citation data is frequently used to measure scientific impact \cite{garfield1972citation,garfield1964science}. Most citation indicators are based on citation counts -- Journal Impact Factor \cite{garfield2006history}, $h$-index \cite{hirsch2005index}, Eigenfactor \cite{fersht2009most}, i-10 index \cite{connor2011google}, c-index \cite{post2018c},  etc. Many variations and adaptations were proposed to compensate the drawbacks of these indices. For instance, $m$-quotient \cite{hirsch2005index,thompson2009new} attempts to eliminate the bias of $h$-index towards older researchers/articles. $g$-index \cite{egghe2006improvement} and $e$-index \cite{zhang2009index} were proposed to overcome bias again authors with heavily cited articles.
We proposed $C^3$-index \cite{Pradhan2017} to resolve ties while ranking medium-cited and low-cited authors by $h$-index. 
Even though so many variations of h-index were proposed in the literature, \citet{bornmann2011multilevel} concluded that most of them are redundant by showing a mean correlation coefficient of $0.8$-$0.9$ between h-index and its 37 alternatives. Few attempts were made to quantify the contribution of individual authors in multi-authored publications \cite{ioannidis2008measuring,howard1983research,romanovsky2012revised,lee2009use}.  

To measure the impact of a scientific article, raw citation count has by far been the most accepted and well studied metric \cite{redner1998popular,radicchi2008universality}. However, many studies confronted with different views against citation count, giving rise to several alternatives such as {\em influmetrics} \cite{bollen2006mapping}, {\em webometrics} \cite{almind1997informetric}, {\em usage metrics} \cite{kurtz2011usage}, {\em altmetrics} \cite{haustein2014coverage}, etc.  \citet{chakraborty2015categorization} showed that the change in yearly citation count of articles published in journals is different from articles published in conferences. Even the evolution of yearly citation count of papers varies across disciplines \cite{chakraborty2018universal,ravenscroft2017measuring}. This further raises a new proposition of designing domain-specific impact measurement metrics.

\subsection{Understanding Citation Expansion} Despite such a vast literature on the use of citation count for assessing the quality of scientific community, the evolution of citation structure has remained largely unexplored. There have been a few recent studies which attempted to understand the organization of citations around a scientific entity (paper, author, venue etc.), particularly focusing on the topology of the graph constructed from the induced subgraph of papers citing the seed paper.  \citet{waumans2016genealogical} took an evolutionary perspective to propose an algorithm for constructing genealogical trees of scientific papers on the basis of their citation count evolution over time. This is useful to trace the evolution of certain concepts proposed in the seed paper. \citet{Singh:201} developed a relay-linking model for prominence and obsolescence
 to include the factors like aging, decline etc. in the evolving citation network.
 \citet{min2018innovation} characterized the citation diffusion process using a classic marketing model \cite{bass1969new} and noticed some interesting patterns in the spread of scientific ideas.  Inspired by information cascade modeling in online social networks \cite{Cheng:2014},   they \cite{min2017quantifying} further made an attempt to study the behavior of citation cascade. They concluded that the average depth of the cascade tends to be influenced by both the lifespan and the whole volume of scientific literature.  \citet{huang2018number} and \citet{abs-1806-00089} argued that citation cascade helps us better understand the citation impact of a scientific publication. They empirically showed that most of the properties of the cascade graph (such as cascade size, edge count, in-degree, and out-degree) follow typical power law distributions; however cascade depth follows exponential distribution. 
 
\subsection{Differences from Previous Literature} Although recent studies \cite{min2017quantifying,huang2018number,abs-1806-00089} argued that there is a need to explore the organization of citations (followup papers) around a seed paper in order to measure better scientific impact, no one quantitatively studied the impact of such network. We are the first to propose an impact measurement metric, called `Influence Dispersion Index' (Section \ref{sec:idi}) which is derived upon converting a rooted citation network to a sparse representation, called `influence dispersion tree' (IDT) (Section \ref{sec:approach}). We show how an optimal orientation of CDT (in terms of its depth and breadth) helps in gaining more impact, which may not be explained by simple citation count. Moreover, the construction of IDT is unique and different from the citation cascade graph proposed earlier \cite{min2017quantifying,huang2018number,abs-1806-00089} (see Section \ref{sec:approach} for more details).

\section{Influence Dispersion Tree (IDT)}
\label{sec:approach}

In this section, we first develop and define the concept of Influence Dispersion Tree of a scholarly paper and describe some of the properties of IDTs. We then develop a simple measure to estimate the \textit{influence} of a scholarly paper given its IDT. 

\subsection{Constructing IDT}
\label{sec:idt-construction}

\begin{figure*}[ht!]
    \centering
    \includegraphics[width=0.9\textwidth]{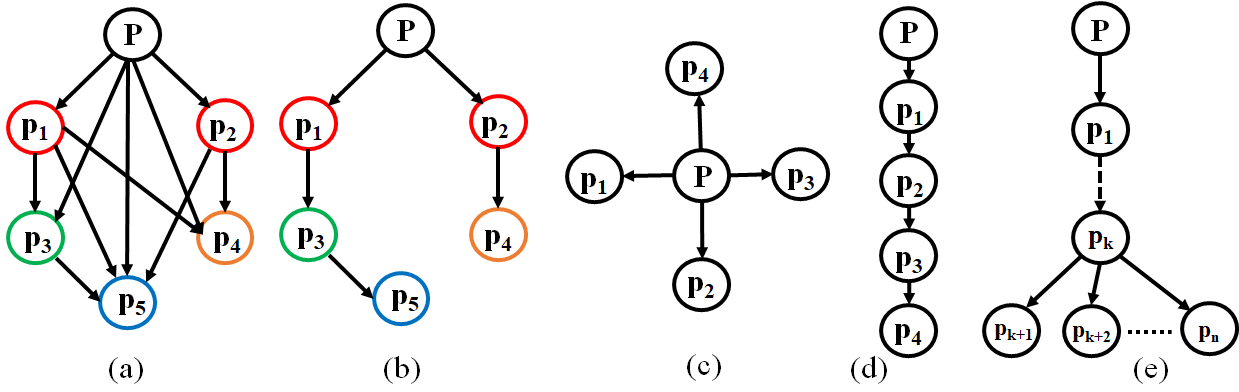}
    \caption{(a)-(b) Illustration of the construction of (b) IDT from (a) IDG of paper $P$. Papers in red only cite $P$; Papers in green cite $P$ and one other paper in the graph; blue paper cites $P$ and more than one other paper in the graph. In case of yellow paper, a tie-breaking occurs due the equal possibility of $p_4$ being connected from $p_1$ and $p_2$ in order to maximize the depth of IDT. Tie-breaking is resolved by randomly connecting $p_4$ from $p_2$ in IDT. (c)-(d) Two corner cases to illustrate the lower bound -- minimum and maximum number of leaf nodes. (e) A configuration of a $P$-rooted IDT with $(n)$ non-root nodes that results in maximum IDI value.}
    \label{fig:toyexample}
\end{figure*}

Let us consider a scholarly paper $P$ and let $\mathcal{C}_P = \{p_1,p_2,\dots,p_n\}$ be the set of papers citing $P$. We assume that $P$ has {\em equally  and directly} influenced each and every paper in $\mathcal{C}_P$.\footnote{Although previous studies~\citep{chakraborty2016all,zhu2015measuring} have found that a paper has a varying amount of influence on its citing papers, 
it is a common practice to assume uniform influence for simplification (e.g., in computing impact factors, h-index~\citep{hindex}, etc.) and is the assumption we also make.}

\begin{definition}\ [{\bf Influence Dispersion Graph}]
The Influence Dispersion Graph (IDG) of the paper $P$ is a directed and rooted graph $\mathcal{G}_P(\mathcal{V}_P,\mathcal{E}_P)$ with $\mathcal{V}_P = \mathcal{C}_P \cup \{P\}$ as the vertex set and $P$ as the root. The edge set $\mathcal{E}_P$ consists of edges of the form $\{p_u \rightarrow p_v\}$ such that $p_u \in \mathcal{V}_P, p_v \in \mathcal{C}_P$ and $p_v$ cites $p_u$. 
\end{definition}

Figure \ref{fig:toyexample}(a) shows an illustration of an IDG for the paper $P$ and its citing paper set $\{p_1,p_2,p_3,p_4,p_5\}$. Observe that the IDG of paper $P$ is the same as the induced subgraph of the larger citation graph consisting of $P$ and  all its citing papers, and with edges in the opposite direction to indicate the propagation of influence from the cited paper to the citing paper. Further, note that the construction of an IDG is similar to that of citation cascades~\citep{huang2018number,min2018innovation} with the fundamental difference that the IDG is restricted strictly to the one-hop citation neighborhood of $P$ (i.e., papers that are directly influenced by $P$) as opposed to the citation cascade that considers higher order citation neighborhoods as well (i.e., papers indirectly influenced by $P$). Thus, an IDG only considers followup papers that are {\em directly influenced} by a given paper. If $p_1$ cites $P$; and $p_2$ cites $p_1$ but  not $P$, it is not always clear if $p_2$ is influenced by both $P$ and $p_1$, or solely by $p_1$.  
Thus, we make the stricter and unambiguous choice by selecting only $p_1$ to be included in the IDG. Though variants of IDG could be constructed by adding additional followup papers, we believe that the major conclusions drawn from the paper will remain valid owing to the stricter and unambiguous process of constructing the IDG.

Next, to further analyze and study the influence of paper $P$ on its citing papers, we derive the \emph{Influence Dispersion Tree} (IDT) of $P$ from its IDG. 
A tree structure, by definition, provides a hierarchical view of the influence $P$ exerts on its citing papers and provides an easy to understand representation to study the relation between $P$ and its citing papers. The IDT of paper $P$ is a directed and rooted tree $\mathcal{T}_P = \{\mathcal{V}_P,\mathcal{E'}_P\}$ with $P$ as the root. The vertex set is the same as that of IDG of $P$ and the edge set $\mathcal{E'}_P \subset \mathcal{E}_P $ is derived from the edge set of IDG as described next.

Note that a paper $p_v \in \mathcal{C}_P$ can cite more than one paper in $\mathcal{V}_P$, giving rise to the following three possibilities:

\begin{enumerate}
	\item $p_v$ cites only the root paper $P$. In this case, we add the edge $P \rightarrow p_v$ creating a new branch in the tree emanating from root node (e.g., edges $P \rightarrow p_1$ and $P \rightarrow p_2$ in Fig. \ref{fig:toyexample}(b)). 
	
	\item $p_v$ cites the root paper $P$  and $p_u \in \mathcal{C}_P \setminus \{p_v\}$. In this case, we say that $p_v$ is influenced by $P$ as well as $p_u$. There are two possible edges here: $P \rightarrow p_v$ and $p_u \rightarrow p_v$. However, since $p_u$ is also influenced by $P$, the edge $p_u \rightarrow p_v$ indirectly captures this influence that $P$ has on $p_v$.  We therefore  retain only the edge $p_u \rightarrow p_v$. This choice leads to addition of a new leaf node in IDT capturing the chain of impact starting from $P$ up to the leaf node $p_v$ (e.g., edge $p_1 \rightarrow p_3$ in  Fig. \ref{fig:toyexample}(b)). 
	
\myComment{	\item $p_v$ cites the root paper $P$, as well as a set of other papers $P_u \subseteq \mathcal{C}_P \setminus \{p_v\}$, $|P_u| >= 2$. Note that by definition, each $p \in P_u$ also cites the root paper $P$. The possible edges to add here are $E = \{P \rightarrow p_v\} \cup \{p \rightarrow p_v\}; \forall p \in P_u$. We retain the edge $e$ in $\mathcal{E}_P$ such that
	\begin{equation}
	    e = \argmax_{e' \in E} shortestPathLength(P,p_v)
	    \label{eq:depth-maximization}
	\end{equation}
    Edge $P_3 \rightarrow P_5$ in Fig. \ref{fig:toyexample}(b) is such an edge.
 }   
    \item 
    $p_v$ cites the root paper $P$, as well as a set of other papers $P_u \subseteq \mathcal{C}_P \setminus \{p_v\}$, $|P_u| >= 2$. Note that by definition, each $p \in P_u$ also cites the root paper $P$. The possible edges to add here are $E =  \{\{p \rightarrow p_v\}; \forall p \in P_u$\}. We add the edge $e$ to $\mathcal{E}'_P$ such that $e = {p \rightarrow p_v}$ where
	\begin{equation}
	    p = \argmax_{p' \in P_u} \;  shortestPathLength(P,p') 
	    \label{eq:depth-maximization}
	\end{equation}
    Edge $P_3 \rightarrow P_5$ in Fig. \ref{fig:toyexample}(b) is such an edge.
    
 \myComment{   
    \item \note{$p_v$ cites the root paper $P$, as well as a set of other papers $P_u \subseteq \mathcal{C}_P \setminus \{p_v\}$, $|P_u| >= 2$. Note that by definition, each $p \in P_u$ also cites the root paper $P$. The possible edges to add here are $E = \{\{p \rightarrow p_v\}; \forall p \in P_u$\}. We do not consider the edge $p \rightarrow p_v$ as explained in the previous point. We retain the edge $e$ in $\mathcal{E}_P$ such that one of its endpoints $p_u$ is given by
	\begin{equation}
	    p_u = \argmax_{p_u \in P_u} shortestPathLength(P,p_u)
	    \label{eq:depth-maximization}
	\end{equation}
	Thus, $e = p_u \rightarrow p_v$.\\
    Edge $P_3 \rightarrow P_5$ in Fig. \ref{fig:toyexample}(b) is such an edge.}
}
\end{enumerate}

The intuition behind adding edges in this way is to maximize the depth of IDT (if there are more than one edge, and each of which maximizes the depth, then we choose one of them randomly, e.g.,  $p_2 \rightarrow p_4$ in Fig. \ref{fig:toyexample}(b)). The edge construction mechanism is motivated by the citation cascade graph \cite{min2017quantifying,huang2018number}. Upon adding a newly citing paper in $\mathcal{T}_P$, we reconstruct $\mathcal{T}_P$  in such a way that the richness of $P$'s influence to its citing papers is maximally preserved. Richness maximization can be thought of as  maximizing the breadth or the depth of the IDT. We choose the latter one in order to capture the cascading effect into the resultant IDT. 

\begin{definition} [{\bf Influence Dispersion Tree}]
The Influence Dispersion Tree (IDT) of paper $P$ is a tree $\mathcal{T}_P(\mathcal{V}_P, \mathcal{E'}_P)$, whose vertex set $\mathcal{V}_P$ is the union of $P$ and all the papers citing $P$. If a paper $p_v$  cites only $P$ and no other papers in $\mathcal{V}_P$, we add $P \rightarrow p_v$ into the edge set $\mathcal{E}'_P$. If $p_v$ cites other papers $P_u\in \mathcal{V}_P\setminus \{P\}$ along with $P$, we add only one edge $p_x\rightarrow p_v$ (where $p_x\in P_u$) according to Equation~
\ref{eq:depth-maximization}.
\end{definition}

\begin{definition}[{\bf $P$-rooted IDT}]
An IDT is called $P$-rooted IDT when the root node of the tree is $P$.
\end{definition}

Figure~\ref{fig:toyexample} illustrates a toy example of constructing IDT from IDG illustrating all three possible cases of edge connections as discussed above.

\subsection{Properties of IDT}
\label{sec:idt-properties}
In this section, we describe a few important properties of an IDT.\\

\noindent \textbf{(i) Depth:} The depth $d$ of a $P$-rooted IDT is defined as the length of the longest path from the root to the leaf nodes $p_L$ in the tree. 
\begin{equation}
    d = \max_{p_l \in p_L} shortestPathLength(P,p_l)
\end{equation}
where $d$ is the depth of the tree, and $p_L$ is the set of leaf nodes in IDT. The depth of the IDT shown in Figure~\ref{fig:toyexample}(b) is $3$.

The depth of an IDT can be interpreted as the longest chain/series of papers representing a body of work influenced by $P$.

\noindent\textbf{(ii) Breadth:} The breadth $b$ of a $P$-rooted IDT  is defined as the maximum number of nodes at a given level in the tree.
\begin{equation}
    b = \max_{1 \leq l \leq d} |N_l|; \quad N_l:= \{ n \in \mathcal{V}_P | level(n) = l\}
\end{equation}
The breadth of the IDT shown in Figure~\ref{fig:toyexample}(b) is $2$.

\noindent\textbf{(iii) Branch:} A branch $P\rightsquigarrow p_l$ is a path from the root $P$ to the leaf $p_l$ in an IDT.

\noindent\textbf{(iv) Fragmented and Unified Branch:} A branch $P\rightsquigarrow p_l$ is called fragmented when an intermediate node (except root) $p\in P\rightsquigarrow p_l$ becomes a part of another branch $P\rightsquigarrow p_{l'}$. $p$ is then called a {\bf fragment point} of  $P\rightsquigarrow p_l$. In Figure \ref{fig:toyexample}(e), $P\rightsquigarrow p_{k+1}$ is a fragmented branch with $p_k$ as a fragment point. If a branch is not fragmented, it is called as a unified branch.  In Figure \ref{fig:toyexample}(d), $P\rightsquigarrow p_{4}$ is a unified branch.

We now define some properties to describe how depth and breadth of a $P$-rooted IDT are related with $n$ -- the number of citations of $P$ (and the number of non-root nodes in the IDT of $P$).

\begin{lemma}\label{lemma:1}
For a paper $P$ with $n$ citations, the range of the depth $d$ and breadth $b$ of the $P$-rooted IDT is  
        $1 \leq d,b \leq n$.
\end{lemma}

\begin{proof}
The breadth of a $P$-rooted IDT will be maximum (i.e, $n$) when all the $n$ papers cite only the root paper $P$, and there is no citation among these $n$ papers (e.g. Figure~\ref{fig:toyexample}(c)). Likewise, the depth of a $P$-rooted IDT will be maximum (i.e., $n$) when there is a chain of $n$ papers $\{P, p_1, p_2,\cdots,p_{n}\}$ forming a unified branch such that $p_i$ cites $p_{i-1}$, $\forall 2\leq i \leq n$; and $p_i$ also cites $P$, $\forall i$ (e.g., Figure~\ref{fig:toyexample}(d)).
\end{proof}

\begin{lemma}\label{lemma:2}
For a paper $P$ with $n$ citations, the sum of depth $d$ and breadth $b$ of the $P$-rooted IDT is bounded by $n+1$, i.e., 
$d + b \leq n+1$.
\end{lemma}

\begin{proof}
When a new node is added to IDT, there are four possibilities -- breadth increases, depth increases, both increase, and neither increases. The sum of $d$ and $b$ will be maximum when both of them are individually maximum. This will only be possible when all but the root node are involved in either increasing depth or breadth or both. However, we can see that only one node, i.e., the first node attached to the root node, can increase both depth and breadth, and the rest will increase either depth or breadth, but not both. Since the total number of non-root nodes added to IDT are $n$, the sum of $b$ and $d$ can attain a maximum value of $n+1$.  
\end{proof}

\begin{lemma}\label{lemma:3}
For a paper $P$ with $n$ citations and its $P$-rooted IDT, the product of its depth $d$ and breadth $b$ is at least $n$, i.e., 
$db \geq  n$
\end{lemma}

\begin{proof}

$d$ is the maximum length of any branch,  and $b$ is indicative of the number of branches from root to leaf. So, for an IDT whose branching occurs at the root node itself and nowhere else, $db$ represents the number of nodes it can have to maintain its depth as $d$ and breadth as $b$ by adding to those branches which have less than $d$ length. Since $n$ is the number of nodes already present in the IDT, we can say that the number of nodes we can add is $db - n$. Since this quantity is always non-negative as this quantity represents the number of nodes we can add, we have 
\begin{equation}
    db - n \geq 0 \implies db \geq n 
\end{equation}

For those IDTs which have branching in places other than the root i.e., fragmented branches, the nodes which are above the branching nodes, will be counted more than once as they represent multiple root to leaf paths and hence $db$ will give more number of nodes than present in the IDT; hence 
\begin{equation}
    db > n
\end{equation}
Therefore, for both the cases, it is seen that $db \geq n$.
\end{proof}

\subsection{Influence Dispersion Index (IDI)}\label{sec:idi}
Given the IDT of a paper, we define its Influence Dispersion Index (IDI) by the sum of length of all the paths from the root node to all the leaf nodes. 

\begin{definition} [{\bf Influence Dispersion Index}]
The IDI of paper $P$ is defined  as
\begin{equation}\label{eq:idt}
IDI(P) = \sum_{p_l \in p_L} distance(P,p_l)
\end{equation}
where $p_L$ is the set of leaf nodes of the $P$'s IDT $\mathcal{T}_P(\mathcal{V}_P, \mathcal{E}_P)$. 
\end{definition}
 The IDI of $P$ in Figure \ref{fig:toyexample}(b) is $5$.

Intuitively, each leaf node in $P$'s IDT corresponds to a separate branch emanating from the original paper $P$. 
Each branch comprises of the set of papers which are influenced by the root paper in one direction. 
We can interpret IDI as a measure of the \textit{ability} of the paper to distribute its influence. We hypothesize that the more an IDT has unified branch, the more the chance that the influence emanating from $P$  is distributed uniformly.

\begin{figure}
    \centering
    \includegraphics[width=0.9\columnwidth]{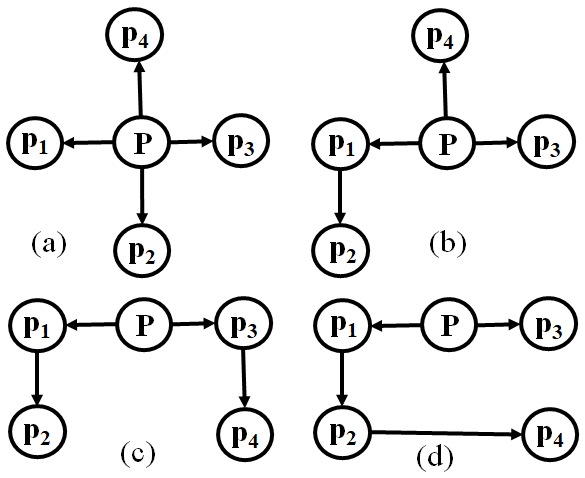}
    \caption{Reconnecting leaf edges of a star IDT (a) to form other configurations.}
    \label{fig:reconfig}
\end{figure}

\subsection{Boundary Conditions of IDI} \label{sec:bound}
\subsubsection{Lower Bound} For a $P$-rooted IDT with $n$ non-root nodes, the minimum value of IDI is $n$. This is because each node (paper) in the tree will be encountered at least once while computing IDI, resulting in the lower bound as $n$. Figures \ref{fig:toyexample}(c) and (d) show two corner cases -- one configuration with the minimum number of leaf nodes (i.e, $1$), and other configuration with the maximum number of leaf nodes (i.e., $n$).  Note that given the size of the IDT, there can be multiple configurations with minimum IDT values. From a star IDT (Figures \ref{fig:toyexample} (c)) if we pick an edge and connect it to any leaf node or the root node, then IDI of the resultant configuration will remain same. In fact, if we keep on repeating the same repairing step, all the resultant configurations will exhibit the same IDI value. In short, during the transformation of a star IDT to a line IDT by reconnecting a leaf edge (an edge whose one end node is a leaf) to another leaf node or to the root node, all the intermediate IDTs will exhibit the same IDI of $n$. Figure \ref{fig:reconfig} shows a toy example of the reconfiguration. We will discuss more in Section \ref{sec:optimaldispersion}.

\subsubsection{Upper Bound:}  In order to maximize the value of IDI, a $P$-rooted IDT should satisfy the following three conditions:
\begin{enumerate}
    \item The number of leaves should be as large as possible.
    \item The length of the branch from root to leaf should be as long as possible.
    \item The number of common nodes in each root-to-leaf branch should be maximized so that each node counter is maximized. 
\end{enumerate}

Subject to the constraint on the number of nodes in the tree (i.e., $n+1$), there is only one structure which can satisfy all the three requirements mentioned above, as shown in Figure \ref{fig:toyexample}(e).

Let IDI of the $P$-rooted IDT with $n$ non-root nodes as shown in Figure \ref{fig:toyexample}(e) be $IDI(P,k)$, where $k$ is the number of nodes forming a chain from $P$ (excluding $P$) and node $p_k$ has $(n-k)$ descendants. Then, $IDI(P,k)$ is determined as follows:
\begin{equation}
    IDI(P,k)  = k(n-k) + (n-k)
\end{equation}

Differentiating it w.r.t to $k$,  we get
\begin{equation}
    \frac{\partial IDI(P,K)}{\partial k} = n - 2k - 1
\end{equation}    

Equating this to $0$ to get the maxima, we get
\begin{equation}
    k = \left \lfloor \frac{n-1}{2} \right \rceil
\end{equation}

This yields the maximum value of IDI as 
    \begin{equation}\label{eq:maxidi}
        IDI(P)^{max} = (1 + \left \lfloor \frac{n-1}{2} \right \rceil) (n - \left \lfloor \frac{n-1}{2} \right \rceil)
    \end{equation}
Therefore, for a $P$-rooted IDT with $n$ non-root nodes, we have the following bounds on its IDI:

\begin{equation}
   \boxed{ n \quad \leq \quad IDI(P) \quad\leq\quad  (1 + \left \lfloor \frac{n-1}{2} \right \rceil) (n - \left \lfloor \frac{n-1}{2} \right \rceil)}
   \label{eq:idi-bounds}
\end{equation}

\noindent

\subsubsection{Relation between $d,b$ and $n$ for Optimal Dispersion}
\label{sec:optimaldispersion}
As discussed above, a paper with a given number of citations $n$, can have differently shaped IDTs, and consequently, very different IDI values. Intuitively, we expect a highly influential paper to have multiple long unified branches, i.e., {\em it should have a high depth value as well as high breadth value}. Thus, we want the IDT of a highly influential paper to have high depth, high breadth, and a tree structure such that the number of non-root nodes are as uniformly distributed in different branches of the trees as possible, indicating significant depth in each branch. Also, recall from Lemma~\ref{lemma:3} that for a given value of $d$ and $b$, the number of nodes in an IDT can not be more than $db$ (i.e., $n \leq db$). This leads us to the following constrained objective function that the IDT in its optimal configuration should satisfy.

\begin{equation}
    \begin{aligned}
        minimize  & \quad (db - n) \\
        \textrm{s.t} &   \quad  d+b \leq n+1 & \textnormal{(from Lemma~\ref{lemma:2})}\\
         \text{and} & \quad db \geq n & \textnormal{(from Lemma~\ref{lemma:3})}
    \end{aligned}
    \label{eq:cost-function}
\end{equation}

\noindent    
This yields an optimal configuration where $d = b = \left\lfloor\sqrt{n}\right\rceil$.

\begin{proof}
    As discussed,  $db$ represents the maximum number of nodes the tree can have by having depth as $d$ and breadth as $b$. The IDT will have maximum number of nodes for a given $d$ and $b$ only when all the branches in the IDT are unified branches. This condition will force the IDT to have all the branches to branch out from the root node. If $k$ is the number of nodes in each unified branch of the optimal tree, and there are $r$ such branches, then the number of nodes in this IDT will be $kr$ (assuming equal length for each branch). Since $k$ and $r$ are equal for an optimal IDT as discussed earlier, we have 
    \begin{equation}
        k^{2} = n \\
        \Rightarrow k = \sqrt{n} 
    \end{equation}
    For IDTs where the nodes are not evenly distributed among an equal number of unified branches with each branch having equal number of nodes (in other words, when the number of non-root nodes is not a perfect square), the corresponding $k$ comes out to be 
    \begin{equation}
        k^{2} = n \\
        \Rightarrow k = \left \lceil \sqrt{n} \right \rceil\ 
    \end{equation}
\end{proof}

\myComment{
 \textbf{Proof:} Since the cost function in \ref{eq:cost-function} is always positive, the minimum value it can attain is $0$. Therefore,
 \begin{equation}
     \quad |d-b| = 0 \\ 
     \Rightarrow d = b 
 \end{equation}
Now we know that for the optimal configuration, both its depth and breadth will be equal. 
Putting $d$ = $b$ and using Lemma \ref{lemma:3}, we get 
\begin{equation}
    d^{2} \geq n 
    \label{eq:d-upperbound}
\end{equation}
and since $d^{2} \geq  0$,  we can say that $d \geq \sqrt{n}$. 
Also using Lemma \ref{lemma:2}, we get 
\[2d \leq n + 1\]
\begin{equation}
    \Rightarrow d \leq \frac{n+1}{2} 
    \label{eq:d-lowerbound}
\end{equation}
Using Equations~\ref{eq:d-upperbound} and \ref{eq:d-lowerbound} and the fact that $d \in I^{+} $ we can clearly see that $d = \left \lceil \sqrt{n} \right \rceil$. Since all the numbers are not perfect squares, we will use ceiling and floor to adjust accordingly, as the quantity $db$ should be as close to $n$ as possible. The rationale behind this, we established that $n$ is the lower bound for $db$, and this lower bound also satisfies Equation~\ref{eq:cost-function}.
Therefore, it implies that for optimal configuration, 
\[d = \left \lceil \sqrt{n} \right \rceil\]
and 
\[b = \left \lfloor \sqrt{n} \right \rfloor\]
}
Figure \ref{fig:optimalconf} illustrates a paper with an optimal configuration where the IDT has an equitable distribution in terms of both depth and breadth, indicating that the paper has influenced multiple branches, and all the influenced branches have grown significantly. Note that the cost function favors configurations where the impact of the paper is maximized both in terms of depth and breadth, and hence, will penalize configurations where there exists a large number of short branches (high $b$, low $d$) or very few long branches (high $d$, low $b$).

\begin{figure}
    \centering
    \includegraphics[width=0.6\columnwidth]{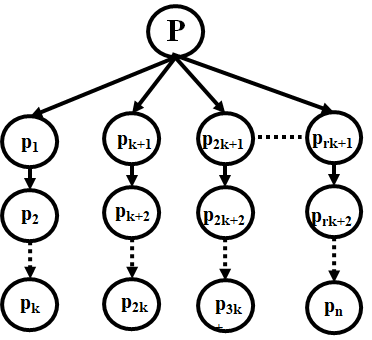}
    \caption{Illustration of an optimal configuration of a $P$-rooted IDT of a paper $P$ with $n$ citations. The depth and breadth of the IDT are same ($k=r=\left\lceil\sqrt{n}\right\rceil$).}
    \label{fig:optimalconf}
\end{figure}

\subsection{IDI as an Influence Measure}
\label{sec:nid}

In this section, we study the potential of IDI as an early predictor of the overall impact and influence of a scholarly article. As discussed before, IDI of a paper $P$ provides a fine-grained view of the influence of $P$ on other papers citing $P$, in terms of the depth and breadth of the IDT. As described in Section~\ref{sec:bound}, for a paper with $n$ citations, there exists an ideal configuration of the IDT that optimizes the influence dispersion of the paper such that it has both high breadth (influenced multiple branches of work) and high depth (significantly deepened each individual branch). With this intuition, we posit that the \emph{closeness} of the actual IDT of a given paper $P$ with $n$ citations, denoted by $\mathcal{T}_P$ to its corresponding ideal IDI with $n$ citations, denoted by $\bar{\mathcal{T}_P}$ can be used as a surrogate measure of influence or impact of paper $P$. We can use any distance metric between two graphs -- such as Graph Edit Distance \cite{gao2010survey}, Gromov-Wasserstein distance \cite{memoli2011gromov} -- to measure the closeness between $\mathcal{T}_P$ and $\bar{\mathcal{T}_P}$. However, all these measures are computationally expensive \cite{gao2010survey}. Therefore, we here use the IDI of each IDT as a proxy for its topological structure and measure the difference between the IDI values of $\mathcal{T}_P$ and $\bar{\mathcal{T}_P}$ (as a replacement of the graph distance). Recall from Section~\ref{sec:bound} that the IDI of an ideal IDT with $n$ non-root nodes is $n$ (which is also the lower bound of an IDT with $n$ internal nodes). 

We define the {\bf Influence Divergence} (ID) of a paper as the difference of the IDI value of its original IDT, IDI(P) and that of its corresponding ideal IDT configuration, $\Bar{IDI}$(P)
\begin{equation}\label{eq:influence}
    ID(P) =  IDI(P) -  \Bar{IDI}(P)
\end{equation}
We further normalize the IDI value using max-min normalization.

\begin{definition} [{\bf Normalized Influence Divergence}]
Normalized Influence Divergence (NID) of a paper $P$ is defined by the difference between the IDI value of its corresponding IDT and the same of its corresponding ideal IDT configuration, $\Bar{IDI}$(P), normalized by the difference between maximum and minimum IDI values of the IDTs with the size as that of $P$'s IDT. Formally, it is written as:

\begin{equation}\label{eq:influence}
    NID(P) =  \frac{IDI(P) -  \Bar{IDI}(P)}{IDI^{max}_{|P|}-IDI^{min}_{|P|}}
\end{equation}
\end{definition}

The normalization is needed to compare two papers with different IDI values. NID ranges between $0$ and $1$.
Clearly, a highly influential paper will have a low $NID(P)$ (i.e., lower deviation from its ideal dispersion index).

\section{Dataset Description}
\label{sec:data}
We used a publicly available dataset of scholarly articles provided by~\citet{chakraborty2018universal}. The dataset contains about $4$ million articles indexed by Microsoft Academic Search (MAS)\footnote{https://academic.microsoft.com/}. For each paper in the dataset, additional metadata such as the title of the paper, its authors and their affiliations, year and venue of publication are also available. The publication years of papers present in the dataset span over half a century allowing us to investigate diverse types of papers in terms of their IDTs.  A unique ID is also assigned to each author and publication venue upon resolving the named-entity disambiguation by MAS itself. We passed the dataset through a series of pre-processing stages such as removing papers that do not have any citation and reference, removing papers that have forward citations (i.e., citing a paper that is published after the citing paper; this may happen due to archiving the paper before publishing it), etc. This filtering resulted in a final set of $3,908,805$  papers. Table~\ref{tab:stat} shows different  statistics of the filtered dataset.

\begin{table}[t]
    \begin{small}
    \centering
    \resizebox{\columnwidth}{!}{
    \begin{tabular}{@{}l| l@{}}
    \toprule
         Number of papers & 3,908,805\\
         Number of unique venues & 5,149\\
         Number of unique authors & 1,186,412\\
         Avg. number of papers per author & 5.21\\
         Avg. number of authors per paper & 2.57\\
         Min. (max.) number of references per paper & 1  (2,432)\\
         Min. (max.) number of citations per paper & 1 (13,102)\\
         \bottomrule
    \end{tabular}
    }
    \caption{Some important statistics about the MAS dataset.}
    \label{tab:stat}
    \end{small}
\end{table}

\begin{table*}[t]
\begin{small}
    \centering
    \begin{tabular}{@{}c L{8cm}cccL{3cm}@{}}
    \toprule
    No. & \multicolumn{1}{c}{Paper} & \# citations & breadth & depth  & \multicolumn{1}{c}{Remark}\\\midrule
    1. & Michael R. Garey and David S. Johnson. 1990. Computers and Intractability; a Guide to the Theory of NP-Completeness. W. H. Freeman \& Co., New York, NY, USA.
 & 13,102 & 4,892 & 34 & A book on the theory of NP-Completeness\\
 
 2. & Cormen, Thomas H., et al. (2001) Introduction to algorithms second edition. & 6777 & 4576 & 8 &  Highly referred text book on Algorithms. \\

 3. & CV. Jacobson. 1988. Congestion avoidance and control. In Symposium proceedings on Communications architectures and protocols (SIGCOMM '88), New York, NY, USA, 314-329. & 2,577 & 259 & 48 & Highly influential paper describing Jacobson's algorithm for control flow in TCP/IP networks  \\
 3. & E. F. Codd. 1970. A relational model of data for large shared data banks. Commun. ACM 13, 6 (June 1970), 377-387. & 2141 & 437 & 42 & Codd's Seminal paper on Relational Databases \\

 \bottomrule   \end{tabular}
    \caption{A  set of representative papers: \#1 and \#2 are the top two papers based on breadth, and \#3 and \#4 are the top two papers based on depth.}
    \label{tab:indicativepaper}
    \vspace{-5mm}
\end{small}
\end{table*}

\section{Empirical Observations}\label{sec:general}
In this section, we report various empirical observations about the IDTs of the papers in our dataset that provide a holistic view of the topological structure of the trees. We also study the how depth and breadth of the IDTs, the IDI and NID values vary with the citation count of the papers.  

\subsection{Structural Properties of IDTs}
\label{sec:properties}

Figure~\ref{fig:depth-breadth-freq} plots the frequency distribution of depth and breadth of the IDTs for all the papers in the dataset. Observe that the values for breadth follow a very long tail distribution with about $75\%$ of papers having a breadth less than or equal to $3$ (note the log-scale on x-axes in Fig.~\ref{fig:b-freq}). On the other hand, the range of the depth values for IDTs is much smaller compared to the range of breadth values. The maximum value of depth is $48$  compared to the maximum breadth of $4,892$. To illustrate the types of papers that achieve very high breadh and depth values, Table~\ref{tab:indicativepaper} lists the top two papers having maximum depth (Papers 1 and 2) and maximum depth (Papers 3 and 4) in our dataset. Note that Papers 1 and 2 are famous Computer Science textbooks resulting in such high breadth values as  most of the citing papers of a book (or survey papers) usually cite the book as a background reference. This may lead to  a large number of short branches in the IDT. On the other hand, Papers 3 and 4 correspond to breakthrough seminal papers --  Paper 3 was among the first to discuss and propose a solution for control flow problem in TCP/IP networks, and Paper 4 is Codd's seminal paper introducing relational databases. These groundbreaking works led to multiple followup papers that build upon these papers resulting in very high depth and relatively low breadth. Also note that even though Papers 3 and 4 have relatively fewer citations than Papers 1 and 2, analyzing the IDT enables us to {\em understand the depth and breadth of the impact of these papers on their citing papers} and measure the influence these papers have had on the fields.

\begin{figure}[t!]
    \centering
	\subfloat[Depth\label{fig:d-freq}]{\includegraphics[width=0.47\columnwidth]{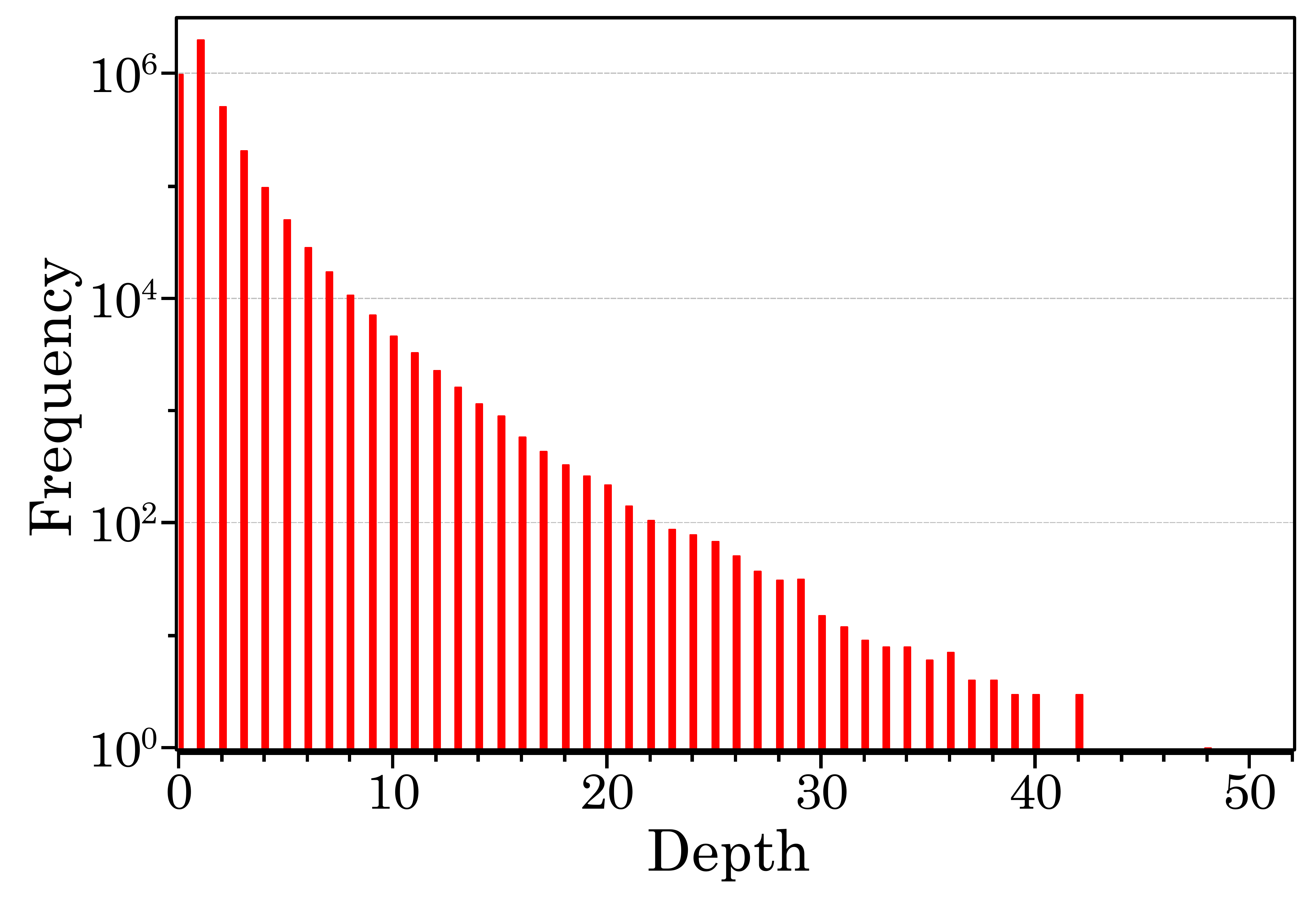}}\quad
	\subfloat[Breadth\label{fig:b-freq}]{\includegraphics[width=0.47\columnwidth]{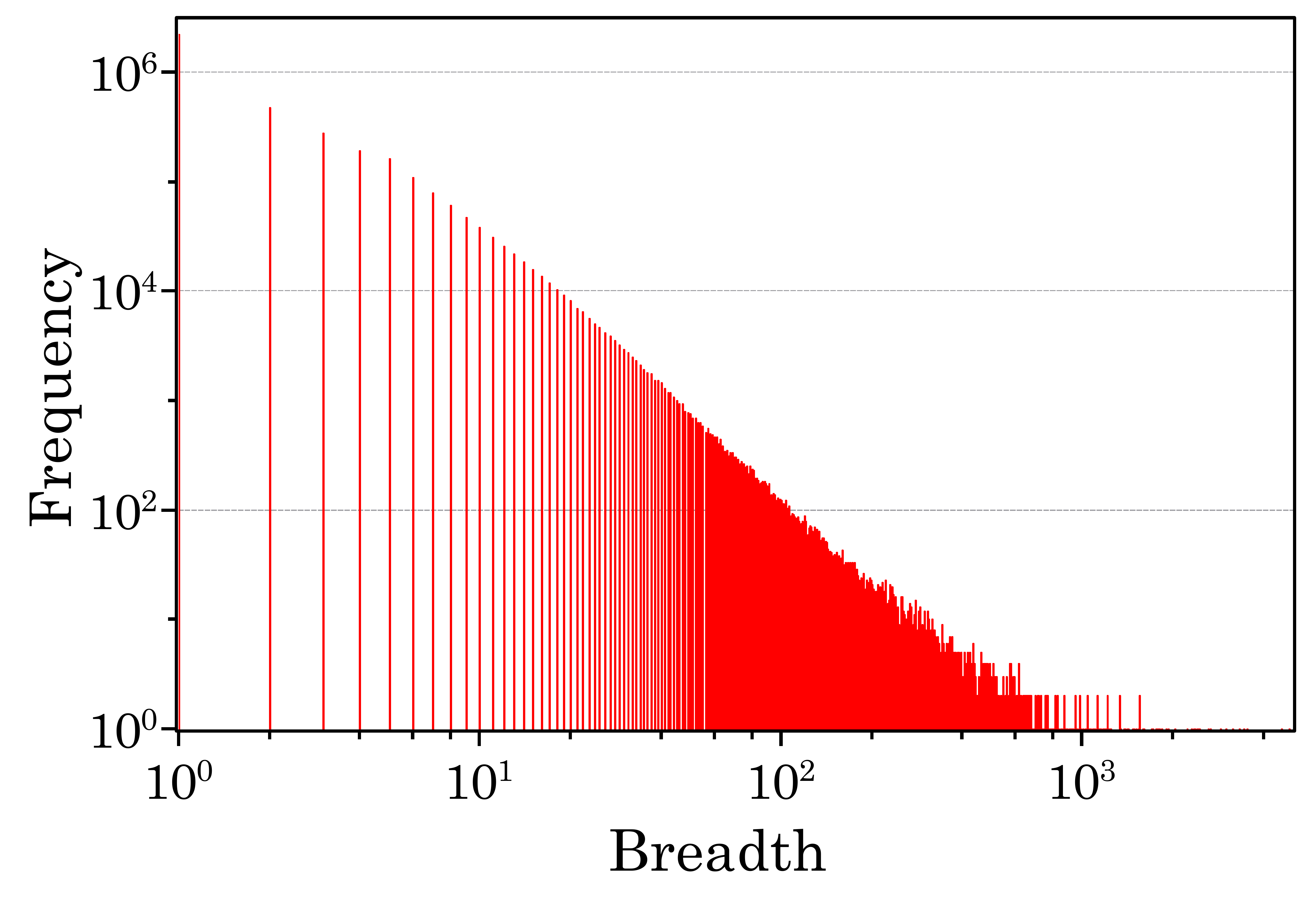}}\quad
	\caption{Frequency distributions for depth (\ref{fig:d-freq}) and breadth (\ref{fig:b-freq}) of IDTs of all the papers in the dataset. The x-axis in the plot for breadth is in logarithmic scale.}
	\label{fig:depth-breadth-freq}
\end{figure}

Figure~\ref{fig:depth-breadth} shows the distribution of breadth and depth with citations (Figures~\ref{fig:b-cit} and \ref{fig:d-cit}, respectively) and the correlation between depth and breadth (Figure~\ref{fig:d-b}). We observe that while breadth is strongly correlated with citation count ($\rho=0.90$), the correlation between depth and citation count is relatively weak ($\rho=0.50$). These observations indicate that increasing citation count often lead to the development of new branches in the IDT of the paper rather than increasing the depth. This happens because most citations to a paper use the cited paper as a background reference (thus gets added to the IDT as a new branch), rather than extending a body of work represented by an already formed branch (increasing the depth). Further, note from Figure~\ref{fig:d-b} that the variation in breadth values reduces with increasing depth. Especially for IDTs with depth greater than $30$, the values of breadth lie in a relatively narrow band (almost all IDTs with depth greater than 30 have breadth less than 300). This is indicative of highly influential papers that have spawned multiple directions of follow-up works and incremental citations correspond to continuation of these independent directions (thus increasing depth).

\begin{figure*}[ht]
    \centering
	\subfloat[Breadth vs. Citations\label{fig:b-cit}]{\includegraphics[width=0.32\textwidth]{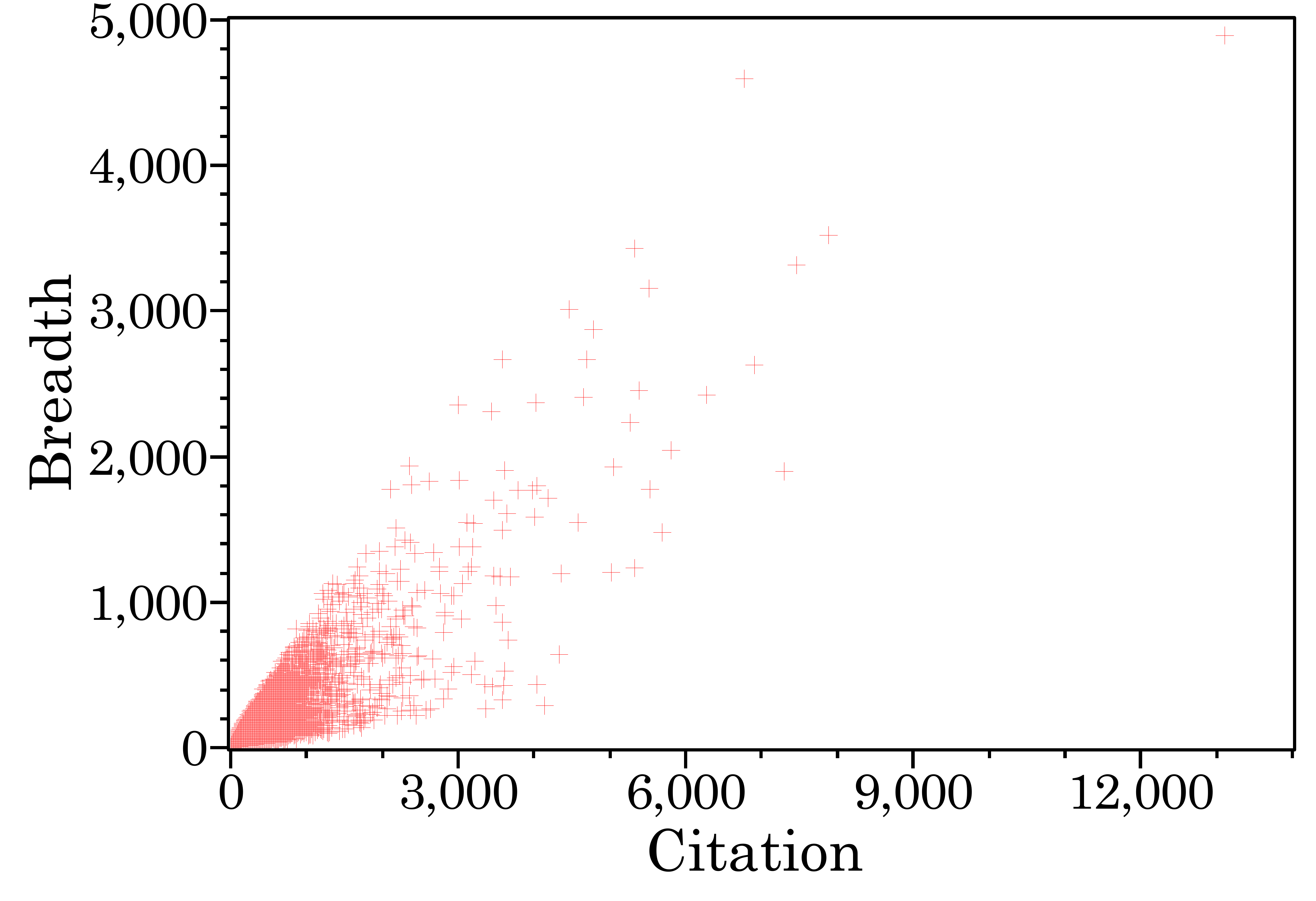}}\quad
	\subfloat[Depth vs. Citations\label{fig:d-cit}]{\includegraphics[width=0.32\textwidth]{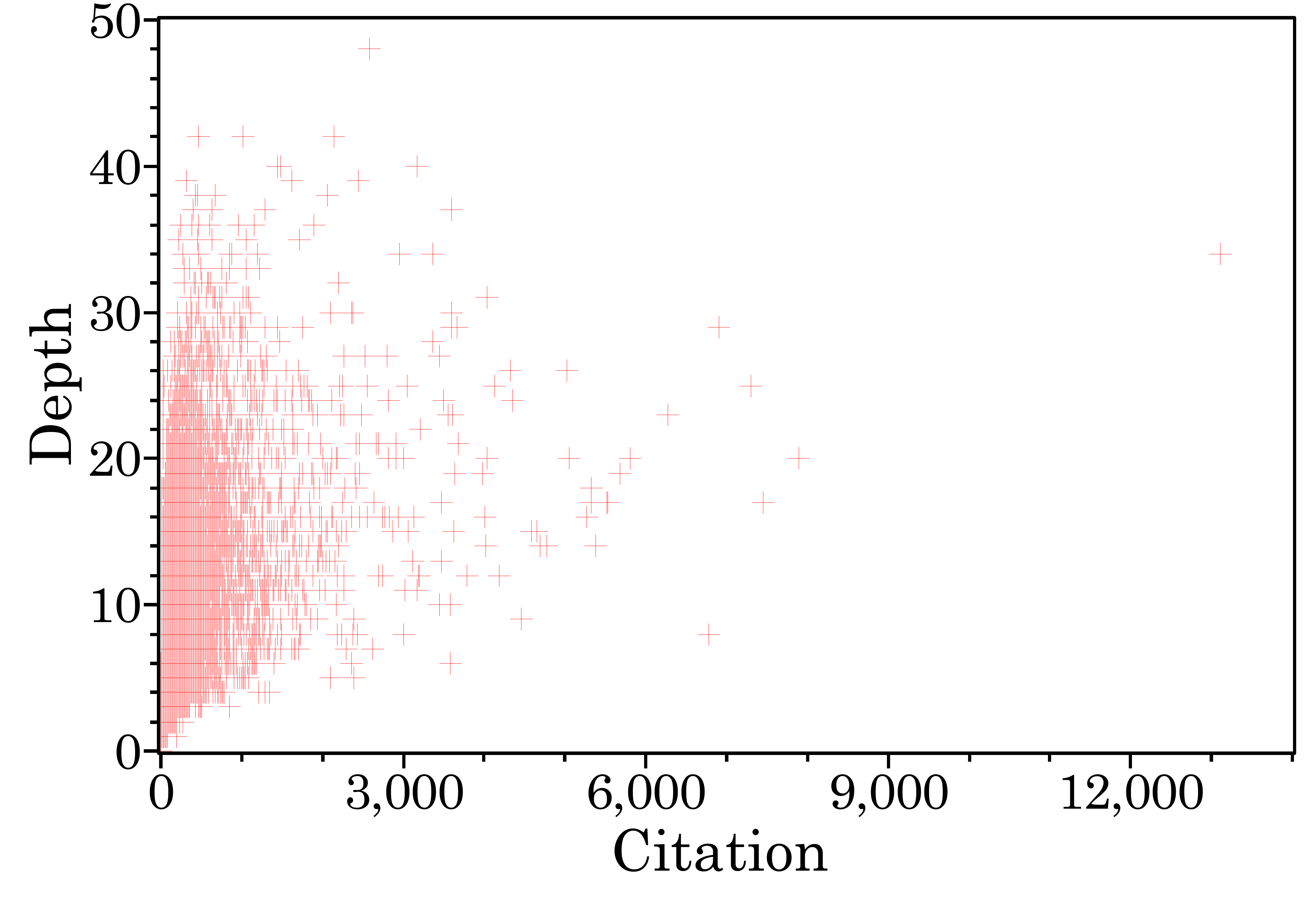}}\quad
	\subfloat[Depth vs. Breadth\label{fig:d-b}]{\includegraphics[width=0.32\textwidth]{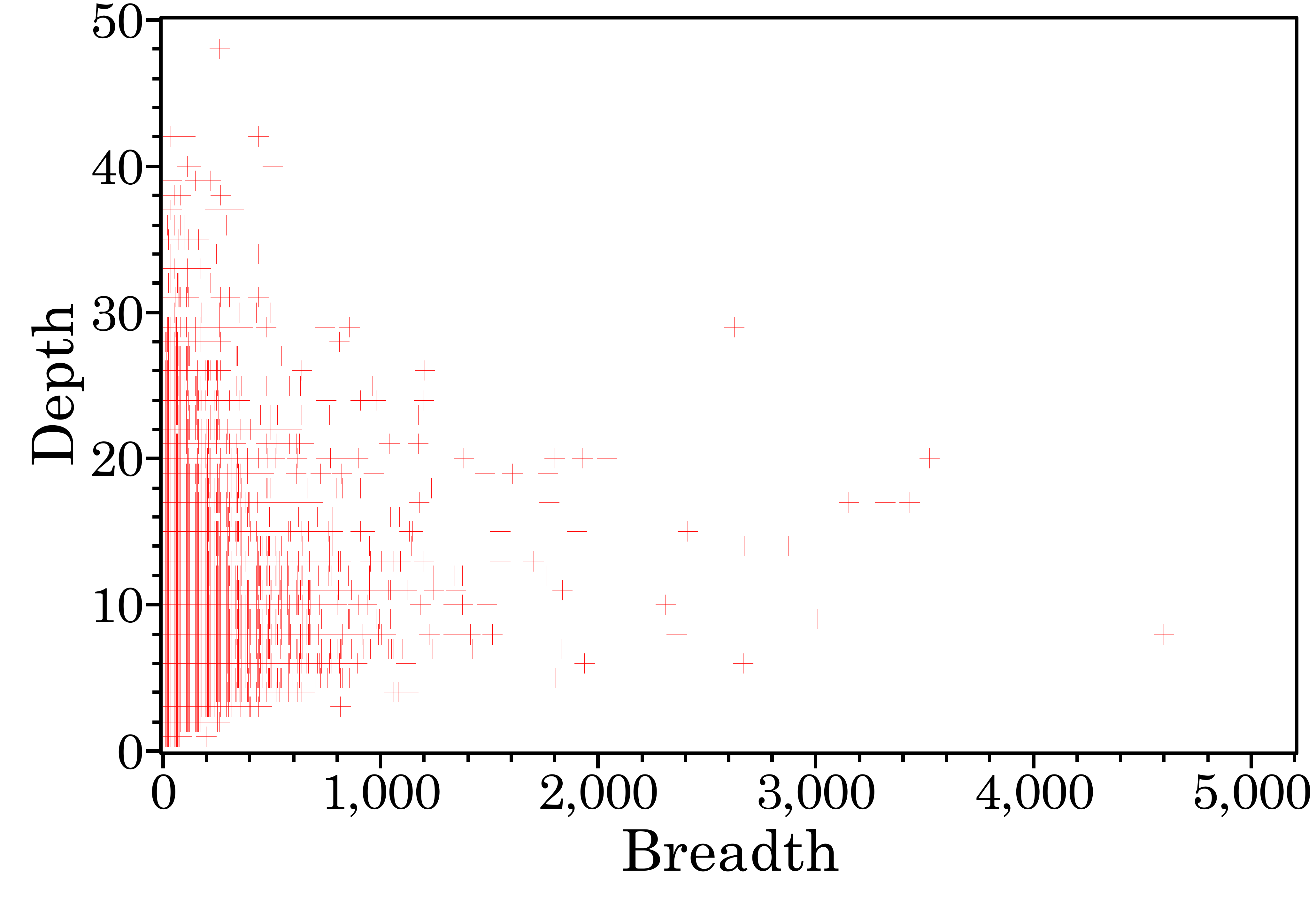}}\quad
	\caption{Scatter plots showing variations of breadth with citations (a), depth with citations (b), and correlation between depth and breadth (c).}
	\label{fig:depth-breadth}
	\vspace{-5mm}
\end{figure*}

\subsection{IDI and NID vs. Citations}
\label{sec:nid-cit}
We now study how the  IDI and NID values vary with the citation counts across multiple papers. Figure~\ref{fig:ideal_nice_dist} shows the scatter plot of IDI and NID values with citations for all the papers in the dataset. We observe that IDI values in general increase with the number of citations of a paper. This is along expected lines as the IDI for a paper is bounded by the number of citations of the paper (Equation~\ref{eq:idi-bounds}). A more interesting observation can be made from the plot for NID values (Figure~\ref{nid-cit}) where we see that in general, the value of  NID decreases with increasing citations -- papers having a high number of citations tend to have very low values of NID. Recall that for a given paper, NID captures how \emph{different} or \textit{far way} the IDI of the given paper is from its corresponding ideal IDT. Thus, highly influential papers tend to have their IDTs close to their ideal IDT configurations (as illustrated by the low NID value). This empirical observation strengthens our hypothesis that \emph{highly influential papers will, in general, lead to considerable amount of followup work (high depth) in multiple directions (high breadth)}.

\begin{figure}[!t]
\centering
    \subfloat[IDI vs. Citations\label{idi-cit}]{\includegraphics[width=0.45\columnwidth]{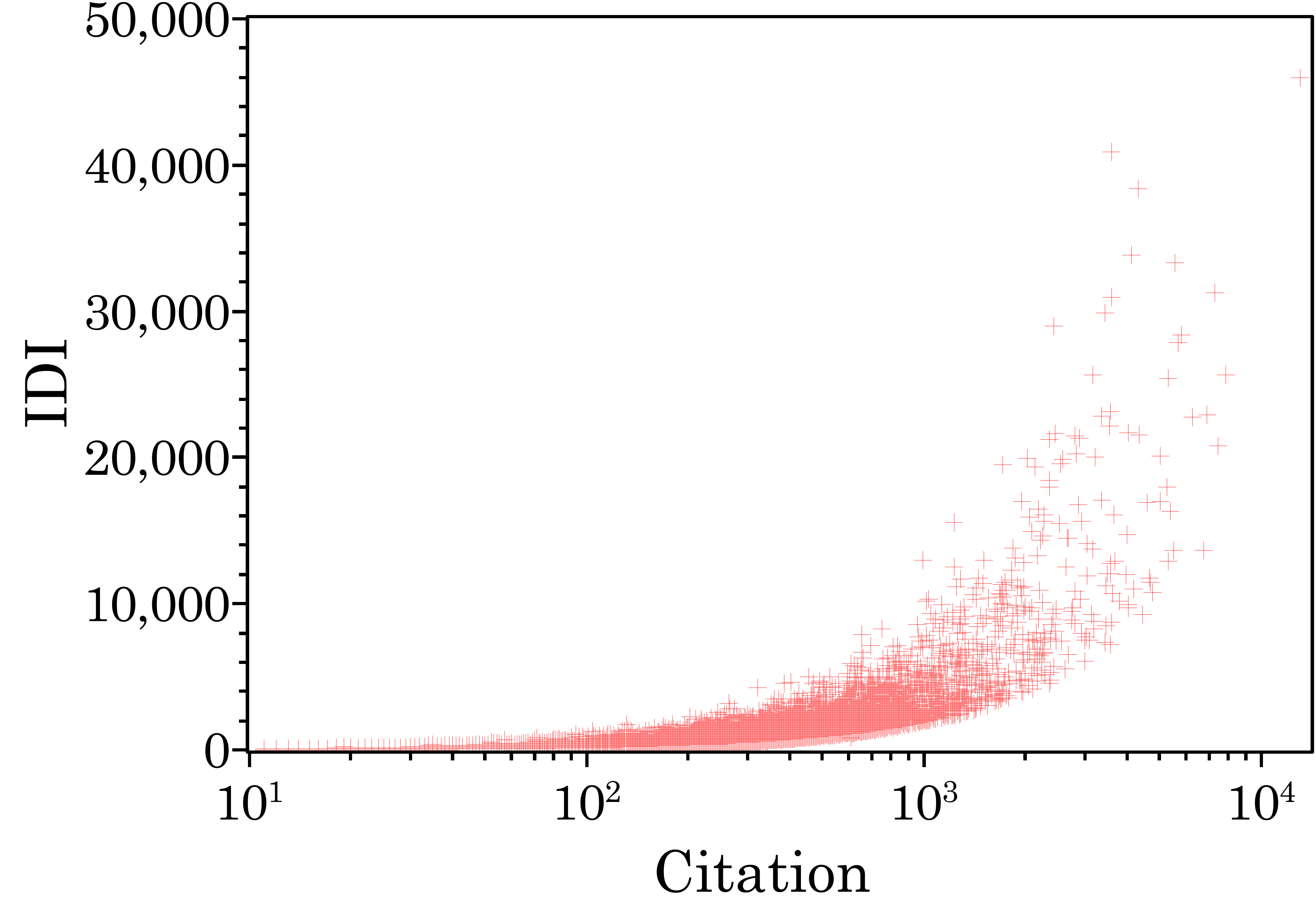}}\quad
	\subfloat[NID vs. Citations\label{nid-cit}]{\includegraphics[width=0.45\columnwidth]{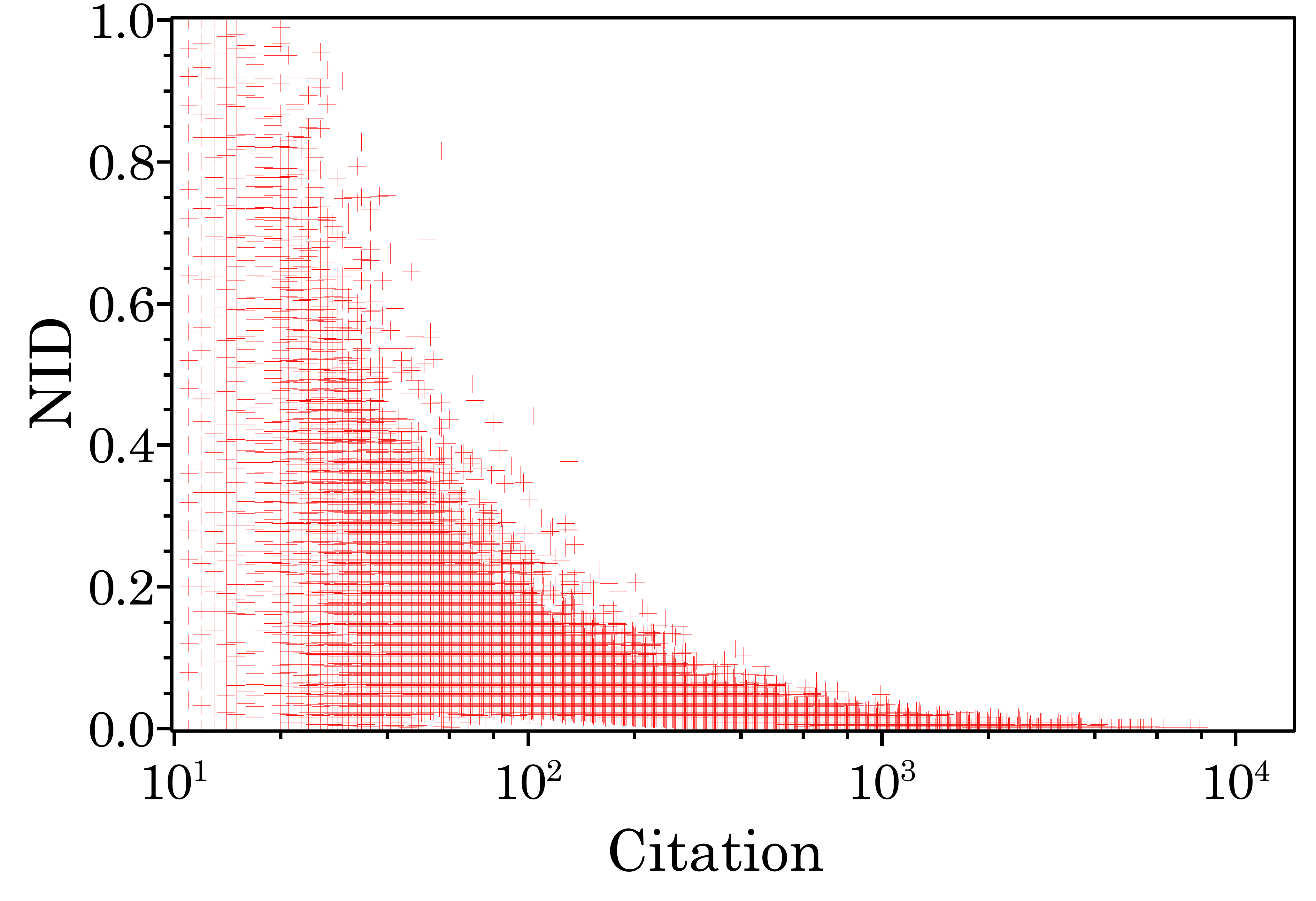}}\quad
    \caption{Scatter plots showing variations of (a) IDI and (b) NID values with citation counts.}
    \label{fig:ideal_nice_dist}
   \end{figure}

\section{NID as an Indicator of Influence}\label{sec:indicator}

As discussed before, we hypothesize that the highly influential papers produce IDTs which would be close to their corresponding ideal configurations. In Section~\ref{sec:nid-cit}, we found that highly-cited papers have very low NID values. Here we ask a complementary question -- \emph{Is low IDI value of a given paper an indicator of its future influence?} In other words, does a paper having its IDT close to the ideal configuration at a given time will be an influential paper in near future? We design two experiments to answer the above question. In Section~\ref{sec:future_citation_prediction}, we study if NID can predict how many citations a paper will get in future. In Section~\ref{sec:tot_papers}, we study if IDI measure can identify highly influential papers -- specifically, papers that have been judged highly influential by the community and have been awarded Test of Time (ToT) awards\footnote{Many conferences and journals award `Test of Time' or `10 year influential paper award' to papers that have had a high impact on their respective fields. These papers are generally selected by a committee of senior researchers.}.

\subsection{Future Citation Prediction through NID}
\label{sec:future_citation_prediction}
Let $P_v$ be the set of papers published in a publication venue $v$ (a conference or a journal). Let  $y_v$ be the year of organization of $v$. Over the next $t$ years, papers in $P_v$ will influence the follow up work and will gather citations accordingly. Let $I(p)$ be an influence measure under consideration. Let $R(v,t,I)$ be the ranked list of papers in $P_v$ ordered by the value of $I(.)$ at $t$. Thus, the top ranked paper in $R(v,t,I)$ is considered to have maximum influence at $t$. If $I(.)$ is able to capture the impact correctly, we expect the papers with high influence scores to have more incremental citations in future compared to papers having low influence scores. Let $C(v,t_1,t_2)$ be the ranked list of papers in $P_v$ ordered by the increase in citations from time $t_1$ to $t_2$. Thus, the papers that received highest fractional increase in citations in the time period $(t_1,t_2)$ will be ranked at the top. Note that we chose fractional increase in citation count rather than absolute count to account for papers that are early risers and receive most of their lifetime citations in first few years after publication \cite{chakraborty2015categorization}. Also, we consider only those papers published in a venue ($v$ here) rather than all the papers in our dataset to nullify the effect of diverse citation dynamics across fields and venues \cite{chakraborty2018universal}. 

Intuitively, if $I(.)$ is a good predictor of a paper's influence, the ranked lists $R(v,t_1,I)$ and $C(v,t_1,t_2)$ should be very similar -- influential papers at time $t_1$ should receive more incremental citations from $t_1$ to $t_2$. Thus, the similarity of the two ranked list could be used as a measure to evaluate the potential of $I(.)$ to be able to capture the influence of papers. We use the Kendall Tau rank distance $\mathcal{K}$ defined below to measure the similarity of the two ranked lists $R(v,t_1,I)$ and $C(v,t_1,t_2)$ as follows.\\

\begin{equation}
    z(v,I) = \mathcal{K}(R(v,t_1,I),C(v,t_1,t_2))
\end{equation}

A lower value of the $z$ score indicates that the two ranked lists are highly similar, that in turn shows that $I(.)$ has high predictive power in forecasting the future incremental citations. We use this framework to evaluate the potential of NID (as a replacement $I(.)$ in this case) as an early predictor of future incremental citations of a paper. We use the number of citations of a paper as a competitor of NID as it is the most common and simplest way of judging the influence of a paper \cite{garfield1972citation,garfield1964science}. First, we group all the papers in our dataset by their venues and compute the values of the influence metrics  (NID and citation count) after five years following the publication year (i.e., $t_1 = 5$). A venue is uniquely defined by the year of publication and the conference/journal series. For example, JCDL 2000 and JCDL 2001 are considered as two separate venues.  We next compute the incremental citations gathered by the papers ten years after the publication ($t_2 = 10$). Note that we only consider venues with the publication year in the range $1995$ and $2000$ because we needed citation information $10$ years after publication (i.e., up to 2010). The coverage of papers published after year $2010$ is relatively sparse in our dataset \cite{chakraborty2018universal}. This filtering resulted in $1,219$ unique venues and $30,556$ papers in total.

With the group of papers published together in a venue and their citation information available, we compute the following three ranked lists:
\begin{enumerate}
    \item $R_{v,c} = R(v,5,c)$; the ranked lists of papers in venue $v$ ordered by their citation counts five years after the publication.
    \item $R_{v,nid} = R(v,5,nid)$; the ranked lists of papers in venue $v$ ordered by their NID scores five years after the publication.
    \item $C_v = C(v,5,10)$; the ranked lists of papers in venue $v$ ordered by the normalized incremental citations received beginning of $5^{\text{th}}$ years after the publication till $10^{\text{th}}$ years after publication.
\end{enumerate}

For each venue $v$, these lists can be used to compute $z(v,NID)$ and $z(v,c)$ -- i.e., the $z$ scores with NID and citation count as influence measures, respectively. For the $1,219$ venues identified as above, the average value of $z$ score using citations and IDI as the influence measure is found to be $0.5125$ and $0.3703$. Thus, on an average, we find that the $Z$ score is lower when using NID as the influence measure compared to that with citation count. In other words, more papers identified as influential by NID received more incremental future citations compared to the papers identified as influential by citation count.

Figure~\ref{fig:z_diff} provides a fine-grained illustration of the difference of $z$ scores achieved by the two influence measures for each of the 1,219 venues. For each venue, we compute the difference of $z$ scores achieved by NID and citation count. We note that for most of the venues, the $z$-score achieved by NID is lower than the $z$-score achieved by the citation count (positive bars). These observations indicate that when compared with raw citation count, NID is a much stronger predictor of the future impact of a scientific paper. As opposed to the raw citation count, the IDT of a paper provides a fine-grained view of the impact of the paper in terms of its depth and breadth as succinctly captured by the IDT of the paper. These results provide compelling evidence for the utility of IDT (and the consequent measures such as IDI and NDI derived from it) for studying the impact of scholarly papers.

\begin{figure}
    \centering
    \includegraphics[width=1\columnwidth]{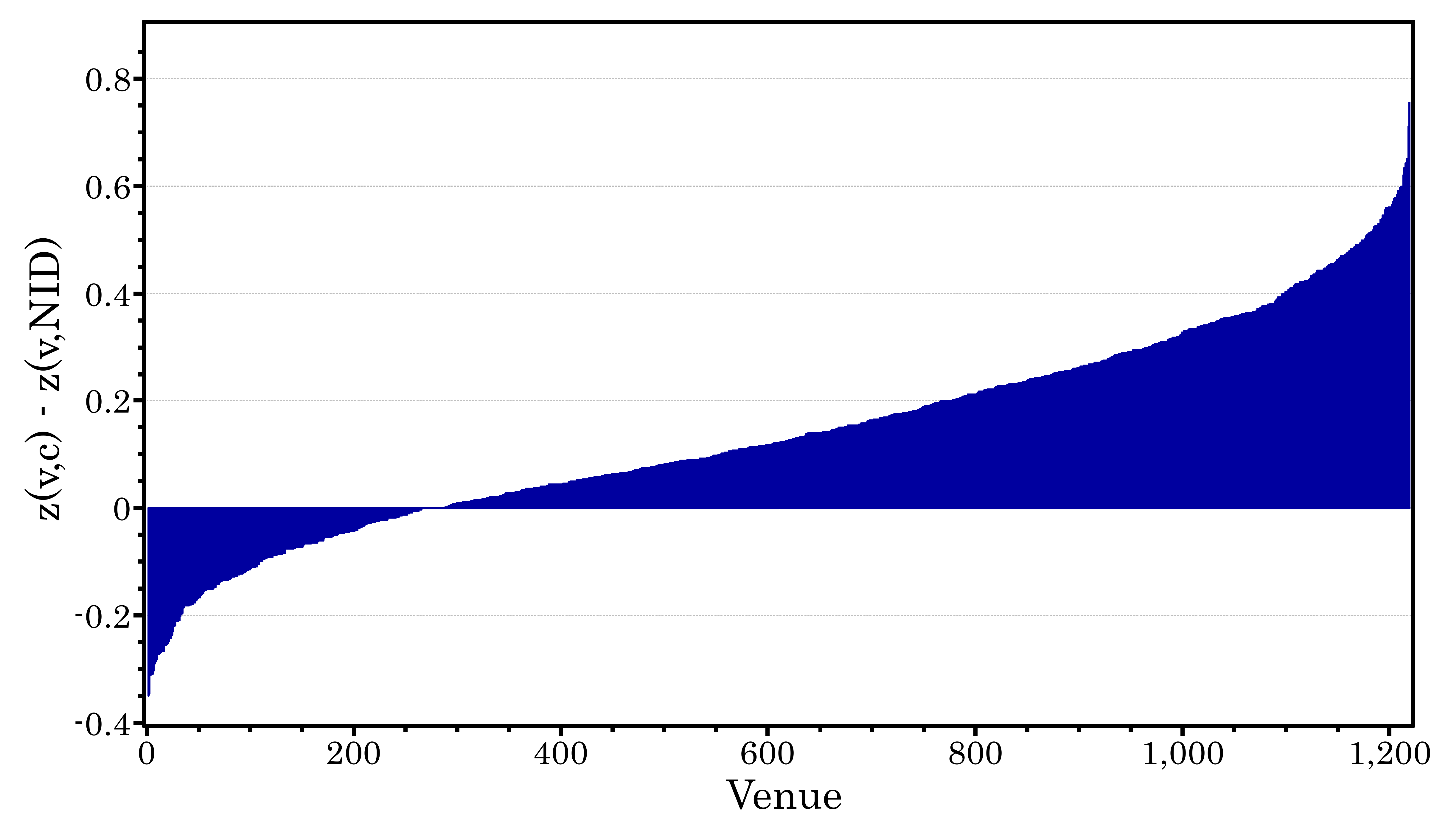}
    \caption{z-scores for venues. Papers in a venue are ranked using NID, number of citations and relative gain in citations. The horizontal axis represents venues ordered by the difference in two z-scores.}
    \label{fig:z_diff}
    
\end{figure}

\subsection{Identifying Test of Time Winners}
\label{sec:tot_papers}

Many conferences recognize highly influential papers that have had a long-lasting impact on the respective field of research. These recognition are awarded in the form of Test of Time (ToT) awards, 10 year Influential Paper Awards, etc. We  manually collected a set of papers that have received the ToT awards by their respective publication venues and obtained a list of 40 such papers (published in conferences like SIGIR, AAAI, ICCV etc.) that are also present in our dataset. 

Let $P$ be a ToT awardee paper that was published in year $y$ at venue $v$. We extracted all the papers from our dataset that were published at venue $v$ in year $y$. We then ordered these papers by their citation count at time $y+10$ (i.e., 10 years after publication) and selected top $5\%$ highest-cited papers (including $P$). We consider these papers to be the major competitor of $P$ to win the TOT award since highly influential papers are expected to achieve a high number of citations\footnote{Many conferences (e.g., SIGIR) nominate top five most cited papers published in a year for the ToT award, in addition to getting nominations from the community.}. We then compute the rank of $P$, denoted by $Rank(P,Cite)$ in this set. Similarly, we compute NID at time $y+10$ for these highly-cited papers and rank them by NID to compute the rank of $P$, denoted by $Rank(P,NID)$. If NID is a better measure of the paper's impact, then we expect $P$ to have a better rank ($1$ being the best outcome, i.e., the top paper) compared to the other papers in the compared set. Figure~\ref{fig:tot_ranking_plots} plots $Rank(P, Cite)$ and $Rank(P, NID$) for each TOT awardee paper $P$. We note that in most of the cases (25 out of 40), the ToT papers are the top-ranked papers by both citation count and NID. 

Interestingly, we also note that in 12 out of 40 cases, the ranks of the ToT awardee papers achieved by NID are lower (better) than the ranks achieved by citation counts. Thus, \emph{the papers judged most influential by the community (by giving TOT award) may not always have the highest citations among all their contemporary papers}. There may be some subjective evaluation criteria that capture the influence a paper has had on the field. The results of this experiment indicate that NID is much better at capturing the influence of a paper -- 33 out of 40 times, the ToT paper achieves rank $1$ when ranked by NID. The overall Mean Reciprocal Rank (MRR) achieved by NID is $0.8771$ compared to an MRR of $0.7712$ achieved by the citation count. Thus, we can consider NID as a much better surrogate measure of influence for a scientific article.

\begin{figure}
    \centering
        \includegraphics[width=0.9\columnwidth]{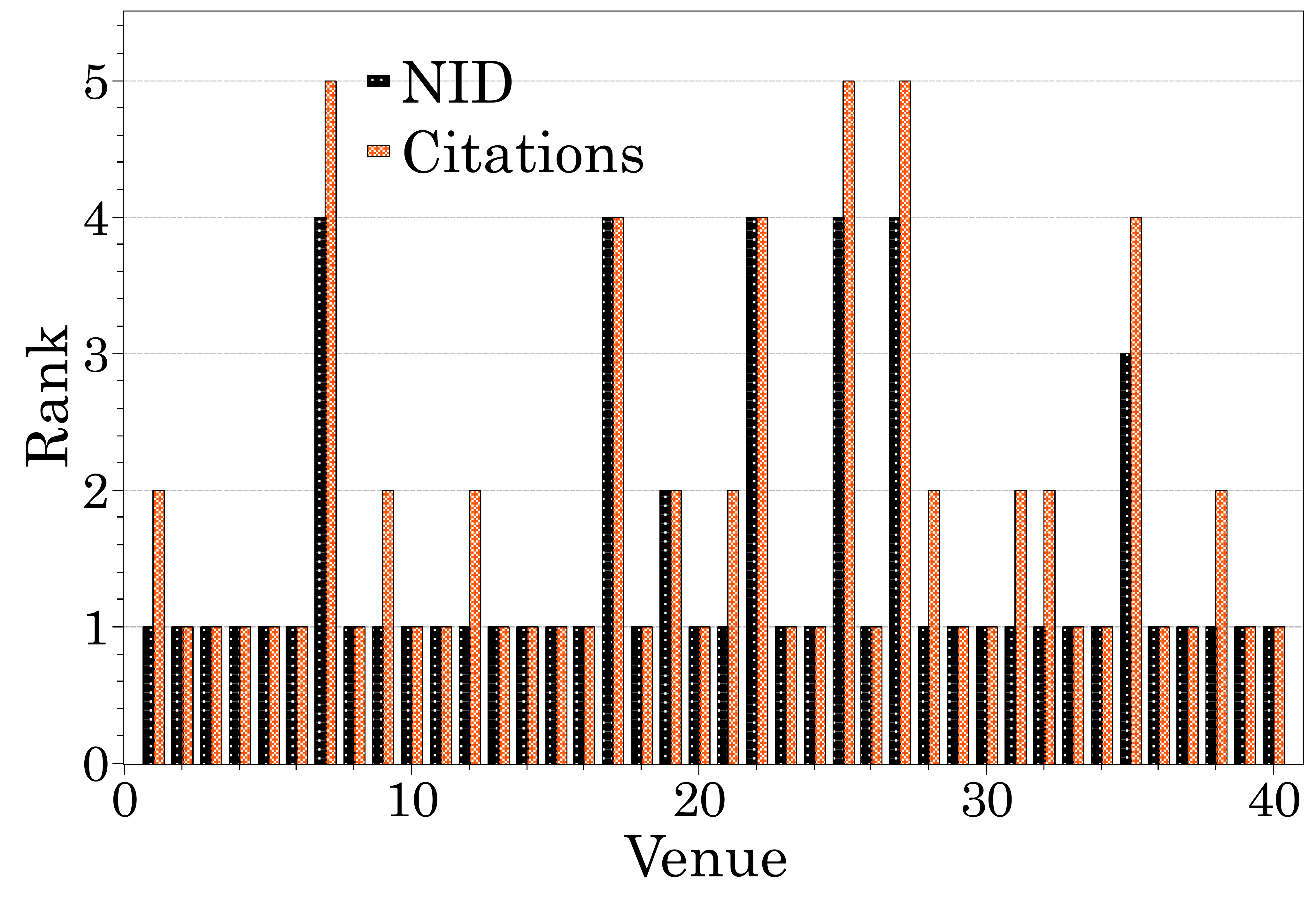}
    \caption{Absolute ranks (based on citation count and NID)  of the ToT papers among their contemporaries. }
    \label{fig:tot_ranking_plots}
\end{figure}

\section{Conclusion}
\label{sec:conclusion}

This paper proposed a novel concept, called `Influence Dispersion Tree' (IDT)  to explore and model the structural information among the followup (citing) papers of a given paper linked through citations. We derive several basic and advanced properties of an IDT to understand their relations with the raw citation count. One striking observation is that with the increase in citation count, the depth of an IDT grows much slower than the breadth. However, as the citation count grows, the IDT of a paper moves closer to its ideal IDT configuration.  
We further proposed a series of metrics to quantify the notion of influence from IDT. Our proposed metric NID turned out to be superior to the raw citation count -- (i) to predict how many new citations a paper is going to receive within a certain time window after publication, (ii) to identify and explain why a paper is recognized by its research community (through various prestigious awards such as Test of Time awards) as highly influential among its contemporaries. 

The conclusion we would like to draw from this paper is -- to understand the contribution of a source paper to its own research field, along with the total number of followup papers of a source paper (i.e., citation count), one should also consider how these followup papers are organized among themselves through citations. A paper can be treated as highly influential only when it has enriched a field equally in both vertical (deepening the knowledge further inside the field) and horizontal (allowing the emergence of new sub-fields) directions. 

\section*{Acknowledgement}
Part of the research was supported by   the Ramanujan Fellowship, Early Career Research Award (SERB, DST), and the Infosys Centre for AI at IIITD.

\balance
\bibliographystyle{ACM-Reference-Format}
\bibliography{refs}

\end{document}